\documentclass[11p,reqno]{amsart}

\usepackage[T1]{fontenc}
\usepackage{graphicx}
\usepackage{amssymb,amsthm,amsmath,mathrsfs,bm,braket,marginnote}
\usepackage{enumerate}
\usepackage{appendix}
\usepackage[colorlinks=true, pdfstartview=FitV, linkcolor=blue, citecolor=blue, urlcolor=blue]{hyperref}
\usepackage{multirow}

\usepackage{pgf}
\usepackage{pgfplots}
\usepackage{tikz}
\usetikzlibrary{arrows,calc}
\usepackage{verbatim}
\usetikzlibrary{decorations.pathreplacing,decorations.pathmorphing}
\usepackage[numbers,sort&compress]{natbib}
\usepackage{dsfont}

\numberwithin{equation}{section}
\newtheorem{theorem}{Theorem}[section]

\newtheorem{remark}[theorem]{Remark}
\newtheorem{corollary}[theorem]{Corollary}

\newtheorem{proposition}[theorem]{Proposition}

 \reversemarginpar

\begin{document}
	
\title[MOPs and non-abelian Volterra lattice]{Matrix orthogonal polynomials, non-abelian Toda lattice and B\"acklund transformation}

\date{}

\dedicatory{}


\author{Shi-Hao Li}
\email{shihao.li@scu.edu.cn}
\address{Department of Mathematics, Sichuan University, Chengdu, 610064, China}

\begin{abstract}
\noindent A connection between matrix orthogonal polynomials and non-abelian integrable lattices is investigated in this paper. The normalization factors of matrix orthogonal polynomials expressed by quasi-determinant are shown to be solutions of non-abelian Toda lattice in semi-discrete and full-discrete cases. Moreover, with a moment modification method, we demonstrate that the B\"acklund transformation of non-abelian Toda given by Popowicz is equivalent to the non-abelian Volterra lattice, whose solutions could be expressed by quasi-determinants as well.
\end{abstract}

\maketitle

\section{Introduction}
Connections between classical integrable systems, orthogonal polynomials and matrix integrals have been largely investigated due to the development of string theory, conformal field theory and random matrix theory \cite{deift03,forrester10}. Toda lattice, as one of the most celebrated integrable systems, has attracted much attention due to its intimate relation to Hermitian matrix model with unitary invariance. Such a connection was revealed by making use of orthogonal polynomials method and  examples could be found in \cite{gerasimov91,adler95}.
If a family of monic orthogonal polynomials $\{P_n(x;t)\}_{n\in\mathbb{N}}$ which satisfy orthogonal relation
\begin{align*}
\langle P_n(x;t), P_m(x;t)\rangle=e^{\phi_n(t)}\delta_{n,m},\quad \langle x^i,x^j\rangle=\int_{\mathbb{R}}x^{i+j}\exp\left(\sum_{k=1}^\infty t_ix^i\right)d\mu(x),
\end{align*}
is introduced,
then it could be shown that the time-dependent normalizations $\{\exp(\phi_n(t))\}_{n\in\mathbb{N}}$ obey the Toda lattice
\begin{align*}
\partial_{t_1}^2\phi_n=e^{\phi_{n+1}-\phi_n}-e^{\phi_n-\phi_{n-1}}.
\end{align*}
Although solutions of Toda lattice were directly expressed by determinant, one noticed that they could be alternatively expressed by the time-dependent partition function of Hermite matrix model with unitary invariance by taking advantage of  Andrei\'ef formula.

Due to the success in Toda lattice and orthogonal polynomials,
there have been many considerations about the interplay between orthogonality and integrability. For example, orthogonal polynomials on the unit circle, skew orthogonal polynomials and corresponding integrable systems have been considered in different literatures \cite{adler99,adler03}, and several novel orthogonality such as Cauchy bi-orthogonal polynomials and partial skew orthogonal polynomials were proposed in the studies of integrable systems \cite{chang182,lundmark05}. Besides, multivariate orthogonal polynomials and multiple orthogonal polynomials have been successfully found within applications to integrable systems \cite{adler09,ari14}. 

In this paper, we mainly focus on matrix orthogonal polynomials and related non-abelian integrable systems. Matrix orthogonal polynomials $\{P_n(x),Q_n(x)\}_{n\in\mathbb{N}}$ are defined by a non-negative matrix-valued weight function $W(x)$, and they obey the following orthogonal relation with respect to the weight
\begin{align}\label{bor1}
\int_\gamma P_n(x)W(x)Q_n(x)dx=H_n\delta_{n,m}.
\end{align}
This is a non-commutative generalization of standard (bi-)orthogonal polynomials and related matrix-valued moment problem could date back to Krein \cite{krein49}. In recent years, there have been a few recent publications concerning applications of matrix-valued orthogonal polynomials into different subjects such as approximation problems on Riemann sphere \cite{bertola21}, quasi-birth-and-death processes \cite{grunbaum08} and random tiling problems \cite{duits20}. Among those, the relation between matrix orthogonal polynomials and integrable systems is always a key topic \cite{alvarez17,branquinho20,cafasso13,ismail19,miranian05}. In particular, in \cite{alvarez17}, a non-abelian integrable systems have already been studied by introducing a Christoffel transformation of matrix orthogonal polynomials. The non-commutative Gauss-Borel decomposition was applied to the moment matrix and some exact formulas are given to those matrix orthogonal polynomials in terms of quasi-determinants, which suggest us that  there should be a non-commutative tau function  theory and some simple Toda-type lattices could be found by making use of matrix orthogonal polynomials and quasi-determinant technique. 

\subsection{Statement of results}
This paper mainly focuses on the interplay between matrix orthogonal polynomials with symmetric measure and non-abelian integrable lattices, including semi/full-discrete non-abelian Toda lattices and a semi-discrete non-abelian Volterra lattice. Although the semi-discrete non-abelian case has been studied in literatures, our method is to connect it with matrix orthogonal polynomials,  to provide a Lax pair together with exact solutions. 
To be precise, we set the weight function $W(x)$ to be symmetric, from which the bi-orthogonal relation \eqref{bor1} reduces to orthogonal relation, i.e. $P_n(x)=Q_n(x)$.
Importantly, those matrix orthogonal polynomials with symmetric weight admit the following three term recurrence relation
\begin{align*}
xP_n(x)=P_{n+1}(x)+a_nP_n(x)+b_nP_{n-1}(x),
\end{align*}
in which $a_n$ and $b_n$ are expressed by quasi-determinants (see Proposition \ref{prop1}).  If we denote $\Phi$ as a column vector composed by $(P_0(x),P_1(x),\cdots)$, then we have the matrix form \footnote{$(P_0,P_1,\cdots)^\top$ means that we take the transpose of the vector $(P_0,P_1,\cdots)$ but keep the elements invariant.}
\begin{align*}
x\Phi(x)=(\Lambda+A\mathbb{I}_p+B\Lambda^{-1})\Phi(x):=\mathcal{L}\Phi(x),\quad \Phi=(P_0(x),P_1(x),\cdots)^
\top,
\end{align*}
where $\Lambda$ is a  block shift operator defined by $
\Lambda\Phi_n=\Phi_{n+1},
$ and $\Phi_n$ is the $n$-th block matrix element of $\Phi$, $A$ and $B$ are certain block diagonal matrices. Therefore, $\mathcal{L}$ is a block tridiagonal matrix (also called Jacobi matrix) with non-commutative elements.
Besides, if one introduces a time-dependent measure $d\mu(x;t)=\exp\left(
\sum_{i=1}^\infty t_ix^i
\right)d\mu(x)$, then evolutions of the wave function $\Phi(x;t)$ could be obtained. The $t_1$-flow corresponds to the non-abelian Toda lattice. By taking the derivative of orthogonal relation, one could get (see Proposition \ref{timeevo})
\begin{align*}
\partial_{t_1} P_n(x;t)=-b_nP_{n-1}(x;t),\quad \text{or} \quad \partial_{t_1}\Phi=-\mathcal{L}_-\Phi,
\end{align*}
where $\mathcal{L}_-$ is the strictly block lower triangular part of $\mathcal{L}$,
and the non-abelian Toda lattice is a result of compatibility condition
\begin{align*}
\partial_t\mathcal{L}=[\mathcal{L},\mathcal{L}_-],
\end{align*}
Moreover, higher order flows lead to other evolutions of wave function and $\partial_{t_n}\Phi=-(\mathcal{L}^n)_-\Phi$ could be constructed, from which a non-commutative Toda hierarchy is given by
\begin{align*}
\partial_{t_n}\mathcal{L}=[\mathcal{L},(\mathcal{L}^n)_-].
\end{align*}

In Section \ref{sec3}, to investigate the connection between full-discrete non-abelian Toda lattice and matrix orthogonal polynomials, we propose a concept of adjacent families of matrix orthogonal polynomials and study corresponding Christoffel and Geronimus transformations. We demonstrate that the Christoffel transformation is a direct result of quasi-determinant identity and takes the form (see Theorem \ref{ctt})
\begin{align*}
P_n^{(\ell)}(x)=xP_{n-1}^{(\ell+2)}(x)-\mathcal{A}_n^\ell P_{n-1}^{(\ell+1)}(x),
\end{align*}
  and the Geronimus one is obtained by orthogonality (see Theorem \ref{gtt})
  \begin{align*}
  xP_n^{(\ell+1)}(x)+P_{n+1}^{(\ell)}(x)+\mathcal{B}_n^\ell P_n^{(\ell)}(x),
  \end{align*}
  where $A_n^\ell$ and $B_n^\ell$ admit quasi-determinant expressions.
 Those discrete spectral transformations are used to construct the discrete Lax pair and thus a full-discrete non-abelian Toda lattice is obtained by the compatibility condition, which is also consistent with three term recurrence relation.

Although adjacent families of matrix orthogonal polynomials provide different solutions of the semi-discrete non-abelian Toda lattice, but how to find a unified auto-B\"acklund transformation is still a mystery for us. Inspired by the scalar case,
we make modifications on the measure such that corresponding polynomials are split into two different families according to the even and odd orders. Moreover, original three term recurrence relations are reduced to
\begin{align*}
xQ_{2n}(x)=Q_{2n+1}(x)+c_nQ_{2n-1}(x),\quad xQ_{2n+1}(x)=Q_{2n+2}(x)+d_nQ_{2n}(x),
\end{align*}
in which coefficients $c_n$ and $d_n$ are also given in terms of quasi-determinants (see Proposition \ref{prop4}). Thus the matrix-valued Lax operator could be expressed by $\mathcal{L}=\Lambda+\Lambda^{-1}\gamma$.
Furthermore, after the introduction of time flows, the non-abelian Volterra lattice
\begin{align*}
\partial_t\mathcal{L}=[\mathcal{L},(\mathcal{L}^2)_-]
\end{align*}
is obtained.
We give an exact quasi-determinant solution to this lattice and demonstrate that this equation links different solutions of non-abelian Toda lattice (see proposition \ref{prop5}). It is also shown that this is a natural discretization of non-abelian KdV equation.

\section{Quasi-determianants, matrix orthogonal polynomials and non-abelian Toda lattice}\label{sec2}

Non-commutative determinants are developed for different special cases such as quantum determinants and so on. In \cite{gelfand05}, the most non-commutative case called the quasi-determinants was introduced for matrices over free division rings. In the same paper, many applications of quasi-determinants are found, including non-commutative continued fractions, matrix orthogonal polynomials and non-abelian Toda lattices and so on. Later on, techniques of quasi-determinants were applied into non-commutative integrable systems such as non-commutative KP, KdV and Painlev\'e equations \cite{etingof98,gilson07,gilson072,li08,retakh10} and combinatorics \cite{di11}. In this section, we shall first give an introduction to quasi-determinant and make a connection with matrix orthogonal polynomials. When it comes to matrix orthogonal polynomials with symmetric measure, a three term recurrence relation could be given and a non-abelian Toda lattice is re-derived with the help of those polynomials. A Lax pair is constructed with wave function given in terms of matrix orthogonal polynomials. Moreover, we remark that this Lax pair coincides with the one in \cite[eq. (8.2)]{o'connell20}, which describes an evolution of a system of particles $X_1,\cdots,X_N$ in the space of invertible $p\times p$ matrices. As quasi-determinant is the main technique in this work, we give a brief introduction here.
\subsection{An introduction to quasi-determinant}
Consider an $n\times n$ matrix $A$ whose entries are defined in a non-commutative ring, $r_i^j$ represents the $i$th row of $A$ with the $j$th element removed, $c^i_j$ the $j$th column with the $i$th element removed and $A^{i,j}$ the submatrix obtained by removing the $i$th row and the $j$th column from $A$, then there are $n^2$ quasi-determinants, denoted as $|A|_{i,j}$ for $i,j=1,\cdots,n$ if all of its inverse $(A^{i,j})^{-1}$ exist, and recursively defined by
\begin{align*}
|A|_{i,j}=a_{i,j}-r_i^j(A^{i,j})^{-1}c_j^i:=
\left|\begin{array}{cc}
A^{i,j}&c_j^i\\
r_i^j&\boxed{a_{i,j}}\end{array}
\right|
,\quad A^{-1}=\left(
|A|_{j,i}^{-1}
\right)_{i,j=1}^n.
\end{align*}
Here the non-commutative ring is supposed to be a $p\times p$ matrix ring, to be in the frame of matrix-valued measure. In \cite[Sec. 1.6]{gelfand05}, quasi-determinants are used to solve a generalization of linear system with non-commutative coefficients.
\begin{proposition}\label{nls}
Let $A=(a_{ij})$ be an $n\times n$ matrix over a $p\times p$ matrix ring. Assume that all quasi-determinants $|A|_{ij}$ are defined and invertible. Then
\begin{align*}
\left\{\begin{array}{c}
a_{11}x_1+\cdots+a_{1n}x_n=\xi_1\\
\cdots\\
a_{n1}x_1+\cdots+a_{nn}x_n=\xi_n
\end{array}
\right.
\end{align*}
for some $x_i\in\mathbb{R}^{p\times p}$ if and only if
\begin{align*}
x_i=\sum_{j=1}^n |A|_{j,i}^{-1}\xi_j, \quad i=1,\cdots, n.
\end{align*}
\end{proposition}
There are some important properties of quasi-determinants (for details, please refer to \cite{krob95,gelfand05,gilson07,li08}). One is the non-commutative Jacobi identity 
\begin{align}\label{ncjacobi}
\left|\begin{array}{ccc}
A&B&C\\
D&f&g\\
E&h&\boxed{i}\end{array}
\right|=\left|\begin{array}{cc}
A&C\\E&\boxed{i}
\end{array}
\right|-\left|\begin{array}{cc}
A&B\\E&\boxed{h}
\end{array}
\right|\left|\begin{array}{cc}
A&B\\D&\boxed{f}
\end{array}
\right|^{-1}\left|\begin{array}{cc}
A&C\\
D&\boxed{g}
\end{array}
\right|,
\end{align}
which is a generalization of Jacobi identity in commutative case, while the latter is one of the most important identities in soliton theory, especially in Hirota's bilinear method. Besides, homological relations in terms of quasi-Pl\"ucker coordinates would be helpful in deriving discrete lattices. By introducing the quasi-Pl\"ucker coordinates, one can easily find the homological relations
\begin{align}\label{homological}
\begin{aligned}
\left|\begin{array}{ccc}
A&B&C\\
D&f&g\\
E&\boxed{h}&i
\end{array}
\right|=\left|\begin{array}{ccc}
A&B&C\\
D&f&g\\
E&h&\boxed{i}
\end{array}
\right|\left|\begin{array}{ccc}
A&B&C\\
D&f&g\\
0&\boxed{0}&1
\end{array}
\right|,\\
\left|\begin{array}{ccc}
A&B&C\\
D&f&\boxed{g}\\
E&h&i
\end{array}
\right|=\left|\begin{array}{ccc}
A&B&0\\
D&f&\boxed{0}\\
E&h&1
\end{array}
\right|
\left|\begin{array}{ccc}
A&B&C\\
D&f&g\\
E&h&\boxed{i}
\end{array}
\right|.
\end{aligned}
\end{align}

According to the viewpoint of classical integrable system (commutative case), determinantal identities acting on tau functions
could be regarded a flow in the Grassmannian, and derivatives of tau function are important in the findings of integrable systems.
Since quasi-determinants work as tau functions in non-abelian integrable systems, 
 it is necessary to introduce derivative formulas for quasi-determinants. Assuming that $A$, $B$, $C$ and $d$ are functions of  dependent $t$, then we have 
\begin{align*}
\left|\begin{array}{cc}
A&B\\
C&\boxed{d}
\end{array}
\right|'=d'-C'A^{-1}B-CA^{-1}B'+CA^{-1}A'A^{-1}B,
\end{align*}
where $'$ means the derivative with respect to $t$. Moreover, if we insert the identity matrix expressed by $
\sum_{k=0}^{n-1} e_k^\top e_k
$
into above formula,
where $e_k$ is a row vector whose $(k+1)$-th element is 1 and others are 0, then we have
\begin{align}\label{de2}
\begin{aligned}
\left|\begin{array}{cc}
A&B\\C&\boxed{d}
\end{array}
\right|'&=\left|\begin{array}{cc}
A&B\\C'&\boxed{d'}
\end{array}
\right|+\sum_{k=1}^n\left|\begin{array}{cc}
A&e_k^\top\\
C&\boxed{0}
\end{array}
\right|\left|\begin{array}{cc}
A&B\\
(A^k)'&\boxed{(B^k)'}
\end{array}
\right|\\
&=\left|\begin{array}{cc}
A&B'\\
C&\boxed{d'}
\end{array}
\right|+\sum_{k=1}^n \left|\begin{array}{cc}
A&(A_k)'\\C&\boxed{(C_k)'}
\end{array}
\right|\left|\begin{array}{cc}
A&B\\e_k&\boxed{0}
\end{array}
\right|,
\end{aligned}
\end{align}
where $A^k$ and $A_k$ denote the $k$th row and column of matrix $A$ respectively.
In particular, as Hankel quasi-determinant plays an important role in non-commutative Toda lattice and 
Painlev\'e equation \cite{gelfand05,retakh10}, we give a derivative formula for Hankel quasi-determinant. A Hankel quasi-determinant is defined by
\begin{align*}
H_n:=\left|\begin{array}{cc}
\Lambda_{n-1}&(\theta_n^n)^\top\\
\theta_n^n&\boxed{m_{2n}}
\end{array}
\right|,\quad \Lambda_{n-1}=(m_{i+j})_{i,j=0}^{n-1},\quad \theta_i^j=(m_i,\cdots,m_{i+j-1}),
\end{align*}
with $m_i=m_i^\top$ required. Obviously, Hankel quasi-determinant has the property $H_n=H_n^\top$. Moreover, by applying derivative formula \eqref{de2}, one could get
\begin{align}\label{derivative}
\begin{aligned}
\partial_t H_n&=\left|\begin{array}{cc}
\Lambda_{n-1}&(\theta_n^n)^\top\\
\theta_{n+1}^n&\boxed{m_{2n+1}}
\end{array}
\right|+\left|\begin{array}{cc}
\Lambda_{n-1}&e_{n-1}^\top\\
\theta_n^n&\boxed{0}
\end{array}
\right|H_n\\
&=\left|\begin{array}{cc}
\Lambda_{n-1}&(\theta_{n+1}^n)^\top\\
\theta_n^n&\boxed{m_{2n+1}}
\end{array}
\right|+H_n\left|\begin{array}{cc}
\Lambda_{n-1}&(\theta_n^n)^\top\\
e_{n-1}&\boxed{0}
\end{array}
\right|.
\end{aligned}
\end{align}

\subsection{Matrix orthogonal polynomials with symmetric measure}
Let's consider a bounded matrix-valued Borel measure $\mu$ on the real line, which is a map from $\mu: (-\infty,\infty)$ to $\mathbb{R}^{p\times p}$. 
This matrix-valued measure is the assignment of a positive semi-definite $p\times p$ matrix $\mu(X)$ to every Borel set $X$ which is countably additive, and it is normalized by assuming $\mu(\mathbb{R})=\mathbb{I}_{p}$, where $\mathbb{I}_p$ is a $p\times p$ identity matrix.
This measure is related to the weight function $M(x)$ in \eqref{bor1} by the Radon-Nikodym theorem, and please refer to \cite{damanik08} for details.

From this matrix-valued measure, an inner product can be defined by
\begin{align*}
\langle f(x),g(x)\rangle:=\int_{\mathbb{R}}f(x)d\mu(x)g^\top (x),
\end{align*}
where $f(x),\,g(x)\in\mathbb{R}^{p\times p}[x]$. A matrix-valued bi-orthogonal polynomials pair $\{P_n(x),Q_n(x)\}_{n\in\mathbb{N}}$ could be defined by this inner product with orthogonal relation
\begin{align*}
\langle P_n(x),Q_m(x)\rangle=H_n\delta_{n,m}
\end{align*}
for some positive definite matrix $H_n$.
This is a non-commutative generalization of standard bi-orthogonal polynomials.
In this paper, we constraint ourselves to a symmetric measure $\mu$ which satisfies $d\mu=d\mu^\top$. With such a constraint, those families of bi-orthogonal polynomials $\{P_n(x)\}_{n\in\mathbb{N}}$ and $\{Q_n(x)\}_{n\in\mathbb{N}}$ coincide with each other, and thus orthogonal relation can be reformulated by $\langle P_n(x),P_m(x)\rangle=H_n\delta_{n,m}$. Equivalently, 
\begin{align}\label{or1}
\int_{\mathbb{R}}P_n(x)d\mu(x)x^m=\int_{\mathbb{R}}x^m d\mu(x)P_n^\top(x)=H_n\delta_{n,m}.
\end{align}
On the other hand, from the matrix-valued measure $d\mu(x)$, a Hankel moment matrix $(m_{i+j})_{i,j\in\mathbb{N}}$ could be introduced with moments
$m_n:=\int_{\mathbb{R}}x^nd\mu(x)\in\mathbb{R}^{p\times p}$.
We always assume the matrix-valued measure is well defined, i.e. moments $\{m_n\}_{n=0,1,\cdots}$ are finite and its corresponding matrix $(m_{i+j})_{i,j=0,1,\cdots}$ are positive definite.

Moreover,
by assuming that $\{P_n(x)\}_{n\in\mathbb{N}}$ be monic polynomials, i.e.
\begin{align*}
P_n(x)&=\mathbb{I}_{p}x^n +a_{n,n-1}x^{n-1}+\cdots+a_{n,0}
\end{align*}
with $a_{n,i}\in\mathbb{R}^{p\times p}$ for $i=0,\cdots,n-1$, 
 then orthogonal relation \eqref{or1} is actually a linear system with non-commutative elements
\begin{align}\label{ls}
\sum_{i=0}^{n-1}a_{n,i}m_{i+j}=-m_{j+n},\quad j=0,\cdots,n-1.
\end{align}
According to Proposition \ref{nls}, those coefficients could be formulated by
\begin{align*}
(a_{n,0},\cdots,a_{n,n-1})=-(m_n,\cdots,m_{2n-1})\left(\begin{array}{cccc}
m_0&\cdots&m_{n-1}\\
\vdots&\cdots&\vdots\\
m_{n-1}&\cdots&m_{2n-2}\end{array}
\right)^{-1},
\end{align*}
from which we know that they could be written as quasi-determinants
\begin{align*}
a_{n,i}=\left|\begin{array}{cc}
\Lambda_{n-1}&e_i^\top\\
\theta_n^n&\boxed{0}
\end{array}
\right|,\quad \Lambda_{n-1}=(m_{i+j})_{i,j=0}^{n-1},\quad \theta^{j}_i=(m_i,\cdots,m_{i+j-1})
\end{align*}
where $e_i^\top$ is a block column vector whose $(i+1)$-th element is $\mathbb{I}_p$ and the others are 0, and thus
\begin{align*}
P_n(x)=\left|\begin{array}{cccc}
m_0&\cdots&m_{n-1}&\mathbb{I}_p\\
\vdots&&\vdots&\vdots\\
m_{n-1}&\cdots&m_{2n-1}&x^{n-1}\mathbb{I}_p\\
m_n&\cdots&m_{2n-1}&\boxed{x^n\mathbb{I}_p}
\end{array}
\right|.
\end{align*}
The algebraic structure of matrix orthogonal polynomials was revealed in \cite[Sec. 8.3]{gelfand95} and in \cite{miranian05}, it was given by the help of Schur complement. At the same time, the quasi-determinant expression is equivalent to a matrix version of LU decomposition in \cite{alvarez17}. Moreover, by inserting the quasi-determinant expression into orthogonal relation \eqref{or1}, we know that the normalization factor $H_n$ can be expressed in terms of Hankel quasi-determinant
\begin{align*}
H_n=\left|\begin{array}{cc}
\Lambda_{n-1}&(\theta_n^n)^\top\\
\theta_n^n&\boxed{m_{2n}}
\end{array}
\right|.
\end{align*}

Following \cite[Thm. 3.1]{sinap96} and \cite[lemma 2.6]{damanik08}, we know that these matrix orthogonal polynomials with symmetric measure admit three term recurrence relations. Moreover, in the next proposition, we demonstrate that coefficients in the three term recurrence relation can be expressed by quasi-determinants.

\begin{proposition}\label{prop1}
 Matrix orthogonal polynomials with symmetric measure $\{P_n(x)\}_{n\in\mathbb{N}}$ have following recurrence relation
\begin{align}\label{sp}
xP_n(x)=P_{n+1}(x)+a_nP_n(x)+b_nP_{n-1}(x),
\end{align}
where recurrence coefficients $a_n$, $b_n$ have quasi-determinantal expressions
\begin{align}\label{qd1}
a_n=\left(
\left|\begin{array}{cc}
\Lambda_{n-1}&(\theta_n^n)^\top\\
\theta_{n+1}^n&\boxed{m_{2n+1}}
\end{array}
\right|+\left|\begin{array}{cc}
\Lambda_{n-1}&e_{n-1}^\top\\
\theta_n^n&\boxed{0}
\end{array}
\right|H_n
\right)H_n^{-1},\quad b_n=H_{n}H_{n-1}^{-1}
\end{align}
\end{proposition}
\begin{proof}
Since $\{P_n(x)\}_{n\in\mathbb{N}}$ act as a basis in the polynomial space, we have the following Fourier expansion
\begin{align*}
xP_n(x)=P_{n+1}(x)+\sum_{k=0}^n \alpha_{n,k}P_k(x)
\end{align*}
with coefficients
\begin{align*}
\alpha_{n,k}=\langle xP_n(x), P_k(x)\rangle\langle P_k(x),P_k(x)\rangle^{-1}.
\end{align*}
From orthogonal relation \eqref{or1}, one knows that when $k<n-1$, 
\begin{align*}
\langle xP_n(x),P_k(x)\rangle=\langle P_n(x),xP_k(x)\rangle=0.
\end{align*}
Moreover, one gets
\begin{align*}
a_n:=\alpha_{n,n}=\langle xP_n, P_n\rangle H_n^{-1},\quad b_n:=\alpha_{n,n-1}=H_n H_{n-1}^{-1},
\end{align*}
and 
\begin{align*}
\langle xP_n,P_n\rangle=\langle x^{n+1},P_n\rangle+a_{n,n-1}\langle x^n, P_n\rangle=\left|\begin{array}{cc}
\Lambda_{n-1}&(\theta_n^n)^\top\\
\theta_{n+1}^n&\boxed{m_{2n+1}}
\end{array}
\right|+\left|\begin{array}{cc}
\Lambda_{n-1}&e_{n-1}^\top\\
\theta_n^n&\boxed{0}
\end{array}
\right|H_n,
\end{align*}
which gives the quasi-determinant expression for coefficients in three term recurrence relation.
\end{proof}
\begin{remark}
In \cite{branquinho20}, the authors tried to consider a connection between matrix orthogonal polynomials and non-abelian integrable lattices. Unfortunately, in their remark 2.6, one of the coefficients is wrongly expressed by quasi-determinant, and their \cite[lemma 4.1]{branquinho20} should coincide with the ``bilinear form'' \eqref{nctoda}, which is clarified later. 
\end{remark}
Therefore, the recurrence relation of matrix orthogonal polynomials can induce the following Jacobi spectral problem
\begin{align*}
x\left(\begin{array}{c}
P_0(x)\\
P_1(x)\\
P_2(x)\\
\vdots\end{array}
\right)=\left(\begin{array}{ccccc}
a_0&\mathbb{I}_p&&&\\
b_1&a_1&\mathbb{I}_p&&\\
&b_2&a_2&\mathbb{I}_p&\\
&&\ddots&\ddots&\ddots\end{array}
\right)\left(\begin{array}{c}
P_0(x)\\
P_1(x)\\
P_2(x)\\
\vdots\end{array}
\right).
\end{align*}
 If we denote $\Phi$ as a column vector composed by $\{P_0(x),P_1(x),\cdots\}$, then above equation is $x\Phi=\mathcal{L}\Phi$, and $\mathcal{L}$ is the corresponding Jacobi matrix with matrix elements.

\subsection{Time evolutions and non-commutative Toda lattice}
To connect matrix orthogonal polynomials with non-abelian Toda lattice, we introduce time flows in the measure such that $\partial_t d\mu(x;t)=xd\mu(x;t)$.
In this case, we know that matrix orthogonal polynomials are also involved in variable $t$ and corresponding matrix moments satisfy evolutions $\partial_t m_n=m_{n+1}$.
Such an example could be interpreted as a symmetric matrix-valued Laguerre weight (in $2\times 2$ case)
\begin{align*}
d\mu(x;t)=\left(
\begin{array}{cc}
x^a&const.\\
 const.&x^b
\end{array}
\right)\exp(xt)\cdot\mathbf{1}_{0<x<\infty}dx,
\end{align*}
with the same constant in off-diagonal part and $\mathbf{1}_{0<x<\infty}$ is the indicator function which is taken as $1$ when $0<x<\infty$ and $0$ otherwise.
More examples could be found at \cite[Sec. 4.1.1]{bertola21} and \cite[Example 4.1]{branquinho20}.

Next we consider evolutions of those polynomials with respect to time variable $t$. 
\begin{proposition}\label{timeevo}
Under the assumption $\partial_t d\mu(x;t)=xd\mu(x;t)$, the matrix orthogonal polynomials $\{P_n(x;t)\}_{n\in\mathbb{N}}$ evolve as
\begin{align}\label{te}
\partial_t P_n(x;t)=-b_nP_{n-1}(x;t),\quad b_n=H_nH_{n-1}^{-1}.
\end{align}
\end{proposition}
\begin{proof}
Since monic polynomials are considered, one could suppose that
\begin{align*}
\partial_t P_n(x;t)=\sum_{k=0}^{n-1}\beta_{n,k}P_k(x;t), \quad \beta_{n,k}=\langle \partial_t P_n(x;t), P_k(x;t)\rangle H_k^{-1}.
\end{align*}
Then by taking the derivative to the orthogonal relation
\begin{align*}
\int_{\mathbb{R}} P_n(x;t)d\mu(x;t)P_k^\top(x;t)=0,\quad \text{for $k<n$},
\end{align*}
one gets
\begin{align*}
\langle \partial_t P_n(x;t), P_k(x;t)\rangle+\langle P_n(x;t),xP_k(x;t)\rangle+\langle P_n(x;t),\partial_t P_k(x;t)\rangle=0.
\end{align*}
Obviously, the last term is equal to zero and thus \begin{align*}
\langle \partial_t P_n(x;t),P_k(x;t)\rangle=-\langle P_n(x;t),xP_k(x;t)\rangle.
\end{align*}
From previous knowledge, we know that $\langle P_n(x;t),xP_k(x;t)\rangle=0$ when $k<n-1$ and 
\begin{align*}
\beta_{n,n-1}=-H_nH_{n-1}^{-1}=-b_k,
\end{align*}
thus completing the proof.
\end{proof}
Therefore, the compatibility condition of the spectral problem (eq. \eqref{sp}) and time evolution (eq. \eqref{te}) gives rise to the non-abelian Toda equation
\begin{align*}
\partial_t a_n=b_{n+1}-b_n,\quad \partial_t b_n=a_nb_n-b_na_{n-1}, 
\end{align*}
which is equivalent to the well known formula \cite[Thm. 9.7.1]{gelfand05}
\begin{align}\label{nctoda}
\partial_t(\partial_t H_n\cdot H_n^{-1})=H_{n+1}H_n^{-1}-H_nH_{n-1}^{-1}
\end{align}
by noting that $a_n=(\partial_t H_n)H_n^{-1}$ according to equation \eqref{derivative} and
$b_n=H_nH_{n-1}^{-1}$. Those transformations look like dependent variable transformations in Hirota's method, and thus we call \eqref{nctoda} a bilinear form of non-abelian Toda lattice.

The Lax pair of this non-commutative Toda lattice is then given in terms of those matrix orthogonal polynomials
\begin{align*}
x\Phi=\mathcal{L}\Phi,\quad \partial_t \Phi=\mathcal{L}_-\Phi,
\end{align*}
where $\mathcal{L}_-$ means the strictly lower triangular part of $\mathcal{L}$ and its compatibility condition is \begin{align*}
\partial_t \mathcal{L}=[\mathcal{L},\mathcal{L}_-]
\end{align*} with commutator $[A,B]=AB-BA$.
\begin{remark}
 In \cite[Sec. 8.1]{o'connell20}, such a Lax pair is connected with following Hamitonian
\begin{align*}
H=\text{tr}\left(
\frac{1}{2}\sum_{i=0}^{N-1} a_i^2-\sum_{i=1}^{N-1}b_i
\right)
\end{align*}
with $a_i=\partial_t X_i(X_i)^{-1}$ and $b_i=X_{i+1}X_i^{-1}$, and those $\{X_i\}_{i=1}^N$ could be considered as being in the space of positive $p\times p$ real symmetric matrices.
\end{remark}
\subsection{Higher order time flows and non-commutative Toda lattice hierarchy}\label{ho}
This part is devoted to higher order time flows with respect to the spectral problem \eqref{sp}. Similar with what we have done in the last part, let's consider infinite time flows in the measure such that
\begin{align*}
d\mu(x;t)=\exp\left(\sum_{k=1}^\infty t_kx^k\right)d\mu(x)
\end{align*}
with $d\mu$ a symmetric matrix measure.
Such an assumption is of great use in deriving Toda lattice hierarchy in scalar case.

We first take $t_2$-flow as an example.
\begin{proposition}
Time-dependent matrix orthogonal polynomials $\{P_n(x;t)\}_{n\in\mathbb{N}}$ have the following evolution with regard to $t_2$-flow
\begin{align*}
\partial_{t_2}P_n(x;t)=-(a_nb_n+b_na_{n-1})P_{n-1}(x;t)-b_nb_{n-1}P_{n-2}(x;t).
\end{align*}
\end{proposition}
\begin{proof}
This proof is similar with the one in Proposition \ref{timeevo}. By assuming that
\begin{align*}
\partial_{t_2}P_n(x;t)=\sum_{k=0}^{n-1}\beta_{n,k}P_k(x;t),\quad \beta_{n,k}=\langle \partial_{t_2}P_n(x;t),P_k(x;t)\rangle H_k^{-1},
\end{align*}
and taking $t_2$-derivative to the orthogonal relation, one gets
\begin{align*}
\langle \partial_{t_2}P_n(x;t),P_k(x;t)\rangle=-\langle P_n(x;t),x^2P_k(x;t)\rangle
\end{align*}
and only $k=n-2$ and $k=n-1$ contribute. By direct computation, one knows that
\begin{align*}
\beta_{n,n-2}=-H_nH_{n-2}^{-1}=-b_nb_{n-1},
\end{align*}
and
\begin{align*}
\beta_{n,n-1}&=-\langle P_n(x;t),x^2P_{n-1}(x;t)\rangle H_{n-1}^{-1}=-\langle x^2P_{n}(x;t),P_{n-1}(x;t)\rangle H_{n-1}^{-1}\\
&=-\langle x(a_nP_n(x;t)+b_nP_{n-1}(x;t)), P_{n-1}(x;t)\rangle H_{n-1}^{-1}\\
&=-a_n\langle xP_n(x;t),P_{n-1}(x;t)\rangle H_{n-1}^{-1}-b_n\langle xP_{n-1}(x;t),P_{n-1}(x;t)\rangle H_{n-1}^{-1}
\end{align*}
and the result is obtained by using \eqref{qd1}.
\end{proof}
We rewrite this formula into matrix form and get
\begin{align*}
\partial_{t_2}\Phi=-(\mathcal{L}^2)_-\Phi,
\end{align*}
then the compatibility condition is $\partial_{t_2}\mathcal{L}=[\mathcal{L},(\mathcal{L}^2)_-]$. Moreover, one could obtain its nonlinear form
\begin{align}\label{t2flow}
\begin{aligned}
\partial_{t_2}a_n&=a_{n+1}b_{n+1}+b_{n+1}a_n-a_nb_n-b_na_{n-1},\\
\partial_{t_2}b_n&=a_n^2b_n-b_na_{n-1}^2+b_{n+1}b_n-b_{n}b_{n-1},
\end{aligned}
\end{align}
which is used to make a comparison with non-abelian Volterra lattice later.

A systematical procedure leads to the $t_n$-flow of the wave function and we state the following proposition.
\begin{proposition}\label{hof}
The higher flows of wave function evolve as
\begin{align*}
\partial_{t_n}\Phi=-(\mathcal{L}^n)_-\Phi.
\end{align*}
\end{proposition}
\begin{proof}
Let's denote 
\begin{align*}
\left( f,g \right)=\left(\begin{array}{ccc}
\langle f_0,g_0\rangle&\langle f_0,g_1\rangle&\cdots\\
\langle f_1,g_0\rangle&\langle f_1,g_1\rangle&\cdots\\
\vdots&\vdots&\ddots
\end{array}
\right),
\end{align*}
then the orthogonal relation \eqref{or1} is equal to $\left(
\Phi,\Phi
\right)=\mathcal{H}$,
in which $\mathcal{H}$ is a block diagonal matrix with elements $(H_0,H_1,\cdots)$. In this notation, we know that $\langle x^n\Phi,\Phi\rangle=\left(\mathcal{L}^n\Phi,\Phi\right)=\mathcal{L}^n\mathcal{H}$ holds for arbitrary $n\in\mathbb{N}$. Taking the $t_n$-derivative to the orthogonal relation, one gets
\begin{align*}
\left(\partial_{t_n}\Phi,\Phi\right)+\left(\Phi,\partial_{t_n}\Phi\right)+\left(
x^n\Phi,\Phi
\right)=\partial_{t_n}\mathcal{H},
\end{align*}
where the first/second part is a strictly block lower/upper triangular matrix respectively and thus we get
\begin{align*}
\left(
\partial_{t_n}\Phi,\Phi
\right)=-(\mathcal{L}^n)_-\left(\Phi,\Phi\right).
\end{align*}
According to the Riesz representation theorem, we get the result.
\end{proof}

The spectral problem \eqref{sp} together with time flows of wave function give rise to the compatibility condition
\begin{align*}
\partial_{t_n}\mathcal{L}=[\mathcal{L},(\mathcal{L}^n)_-].
\end{align*}
To sum up, we state the following proposition.
\begin{proposition}
Introducing the Lax matrix 
\begin{align*}
\mathcal{L}=\left(\begin{array}{ccccc}
a_0&\mathbb{I}_p&&&\\
b_1&a_1&\mathbb{I}_p&&\\
&b_2&a_2&\mathbb{I}_p&\\
&&\ddots&\ddots&\ddots\end{array}
\right),
\end{align*}
then the non-abelian Toda hierarchy could be given by
\begin{align*}
\partial_{t_n}\mathcal{L}=[\mathcal{L},(\mathcal{L}^n)_-]
\end{align*}
with solutions $a_k=(\partial_{t_1}H_k)H_k^{-1}$ and $b_k=H_kH_{k-1}^{-1}$. Moreover, $\{H_k\}_{k=0,1,\cdots}$ are Hankel quasi-determinants whose moments satisfy evolution relations $\partial_{t_n}m_i=m_{n+i}$.
\end{proposition}

\section{Adjacent families of matrix orthogonal polynomials and a discrete non-commutative Toda lattice}\label{sec3}
\subsection{Adjacent families of matrix orthogonal polynomials}
Adjacent families of orthogonal polynomials is a significant part in the theory of orthogonal polynomials, to formulate discrete spectral transformations of matrix orthogonal polynomials. For example, in commutative case, adjacent families of orthogonal polynomials were considered in 
Pad\'e approximation theory \cite{brezinski80} and their applications into discrete integrable systems were considered in \cite{spicer11} for standard orthogonal polynomials and \cite{li20} for skew orthogonal polynomials. Regarding with the matrix orthogonal polynomials with symmetric weight function, one can define an $\ell$-th adjacent family of $\{P_n(x)\}_{n\in\mathbb{N}}$ by $\{P_n^{(\ell)}(x)\}_{n\in\mathbb{N}}$ and
\begin{align*}
P_n^{(\ell)}(x)=\left|\begin{array}{cccc}
m_{\ell}&m_{\ell+1}&\cdots&\mathbb{I}_p\\
\vdots&\vdots&&\vdots\\
m_{\ell+n-1}&m_{\ell+n}&\cdots&x^{n-1}\mathbb{I}_p\\
m_{\ell+n}&m_{\ell+n+1}&\cdots&\boxed{x^n\mathbb{I}_p}
\end{array}
\right|,\quad m_{\ell+k}=\int_{\mathbb{R}} x^kd\mu^{(\ell)}(x),\quad d\mu^{(\ell)}(x)=x^\ell d\mu(x).
\end{align*}
Due to the symmetry of $d\mu(x)$, $d\mu^{(\ell)}(x)$ is still a symmetric measure.
Obviously, adjacent families of matrix orthogonal polynomials $\{P^{(\ell)}_n(x)\}_{n\in\mathbb{N}}$ have orthogonality with $d\mu^{(\ell)}$, and they satisfy
\begin{align*}
\int_{\mathbb{R}}x^k d\mu^{(\ell)}(x) (P_n^{(\ell)}(x))^\top=\int_{\mathbb{R}} x^{k+\ell}d\mu(x)(P_n^{(\ell)}(x))^\top=\langle x^\ell P_n^{(\ell)}, x^k\rangle=0,\quad k=0,\cdots,n-1,
\end{align*}
and 
\begin{align*}
\langle x^\ell P_n^{(\ell)}, x^n\rangle=\left|\begin{array}{ccc}
m_\ell&\cdots&m_{\ell+n}\\
\vdots&&\vdots\\
m_{\ell+n}&\cdots&\boxed{m_{\ell+2n}}
\end{array}
\right|=\left|\begin{array}{cc}
\Lambda^{(\ell)}_{n-1}&(\theta_{n+\ell}^n)^\top\\
\theta_{n+\ell}^n&\boxed{m_{\ell+2n}}
\end{array}
\right|:=H_{n}^{(\ell)}.
\end{align*}
If we introduce an inner product $\langle\cdot,\cdot\rangle_\ell$ with respect to $d\mu^{(\ell)}(x)$, then the above orthogonal relation can be written as
\begin{align}\label{orl}
\langle P_n^{(\ell)}(x),P_k^{(\ell)}(x)\rangle_\ell=H_{n}^{(\ell)}\delta_{n,k}.
\end{align}
Those adjacent families of matrix orthogonal polynomials also admit three term recurrence relations  stated as follows.
\begin{proposition}
The adjacent family of matrix orthogonal polynomials $\{P_n^{(\ell)}(x)\}_{n\in\mathbb{N}}$ satisfy three term recurrence relation
\begin{align}\label{ttrr2}
xP_n^{(\ell)}(x)=P_{n+1}^{(\ell)}(x)+a_n^{(\ell)}P_n^{(\ell)}(x)+b_n^{(\ell)}P_{n-1}^{(\ell)}(x),
\end{align}
where
\begin{align*}
a_n^{(\ell)}=\left(\left|\begin{array}{cc}
\Lambda_{n-1}^{(\ell)}& (\theta_{n+\ell}^n)^\top\\
\theta_{n+\ell+1}^n&\boxed{m_{2n+\ell+1}}
\end{array}\right|+\left|\begin{array}{cc}
\Lambda^{(\ell)}_{n-1}&e_{n-1}^\top\\
\theta_{n+\ell}^n&\boxed{0}\end{array}
\right|H_n^{(\ell)}
\right)(H_n^{(\ell)})^{-1},\quad b_n^{(\ell)}=H_n^{(\ell)}(H_{n-1}^{(\ell)})^{-1}.
\end{align*}
\end{proposition}

Moreover, it is not difficult to follow Proposition \ref{timeevo} to show 
\begin{align*}
\partial_t P_n^{(\ell)}(x)=-b_n^{(\ell)}P_{ n-1}^{(\ell)}(x),
\end{align*}
and a non-abelian Toda equation \eqref{nctoda} could be derived. We can state the following corollary as a direct result.

\begin{corollary}
Each of $\{H_n^{(\ell)}\}_{\ell=0,1,\cdots}$ satisfies the non-abelian Toda equation \eqref{nctoda}
\begin{align*}
\partial_t (\partial_t H^{(\ell)}_n\cdot (H^{(\ell)}_n)^{-1})=H_{n+1}^{(\ell)}(H_n^{(\ell)})^{-1}-H_n^{(\ell)}(H_{n-1}^{(\ell)})^{-1}
\end{align*}
if moments in $H_{n}^{(\ell)}$ satisfy time evolutions $\partial_t m_i=m_{i+1}$.
\end{corollary}

\subsection{Spectral transformations and discrete non-commutative Toda lattice}
After the introduction of adjacent families of matrix orthogonal polynomials, one can consider the following Christoffel transformation.
\begin{theorem}\label{ctt}
The Christoffel transformation of adjacent families of matrix orthogonal polynomials is given by
\begin{align}\label{ct}
P_n^{(\ell)}(x)=xP_{n-1}^{(\ell+2)}(x)-\mathcal{A}_n^\ell P_{n-1}^{(\ell+1)}(x),
\end{align}
where $\mathcal{A}_n^\ell=H_{n-1}^{(\ell+1)}\left((H_{n-1}^{(\ell)})^{-1}-(H_{n-2}^{(\ell+2)})^{-1}\right)$.
\end{theorem}
\begin{proof}
As was shown in \cite{brezinski90}, the Christoffel transformation in scalar case could be directly verified by determinant identity. Therefore, we apply the non-commutative version of Jacobi identity \eqref{ncjacobi} to the adjacent family of matrix orthogonal polynomials and get
\begin{align*}
&\left|\begin{array}{cccc}
m_\ell&\cdots&m_{\ell+n-1}&\mathbb{I}_p\\
\vdots&&\vdots&\vdots\\
m_{\ell+n-1}&\cdots&m_{\ell+2n-2}&x^{n-1}\mathbb{I}_p\\
m_{\ell+n}&\cdots&m_{\ell+2n-1}&\boxed{x^n\mathbb{I}_p}
\end{array}
\right|=\left|\begin{array}{cccc}
m_{\ell+2}&\cdots&m_{\ell+n}&x\mathbb{I}_p\\
\vdots&&\vdots\\
m_{\ell+n}&\cdots&m_{\ell+2n-2}&x^{n-1}\mathbb{I}_p\\
m_{\ell+n+1}&\cdots&m_{\ell+2n-1}&\boxed{x^n\mathbb{I}_p}
\end{array}
\right|\\
&\quad-
\left|\begin{array}{ccc}
m_{\ell+1}&\cdots&{m_{\ell+n}}\\
\vdots&&\vdots\\
\boxed{m_{\ell+n}}&\cdots&m_{\ell+2n-1}
\end{array}
\right|
\left|\begin{array}{ccc}
\boxed{m_\ell}&\cdots&m_{\ell+n-1}\\
\vdots&&\vdots\\
m_{\ell+n-1}&\cdots&m_{\ell+2n-2}
\end{array}
\right|^{-1}
\left|\begin{array}{cccc}
m_{\ell+1}&\cdots&m_{\ell+n-1}&\boxed{\mathbb{I}_p}\\
\vdots&&\vdots&\vdots\\
m_{\ell+n}&\cdots&m_{\ell+2n-2}&x^{n-1}\mathbb{I}_p
\end{array}
\right|.
\end{align*}
By making use of homological relations \eqref{homological} to the last three quasi-determinants, one is led to the result.
\end{proof}
Besides, a Geronimus transformation for those adjacent families could be obtained by simply using the orthogonal relation. We state it as following.
\begin{theorem}\label{gtt}
The Geronimus transformation of matrix orthogonal polynomials is
\begin{align}\label{gt}
xP_n^{(\ell+1)}(x)=P_{n+1}^{(\ell)}(x)+\mathcal{B}_n^\ell P_n^{(\ell)}(x),
\end{align}
where $\mathcal{B}_n^\ell=H_{n}^{(\ell+1)}(H_{n}^{(\ell)})^{-1}$.
\end{theorem}
\begin{proof}
The Geronimus transformation is to find a relation between $\{P_n^{(\ell)}(x)\}_{n\in\mathbb{N}}$ and $\{P_n^{(\ell+1)}(x)\}_{n\in\mathbb{N}}$. Let's assume 
\begin{align*}
xP_n^{(\ell+1)}(x)=P_{n+1}^{(\ell)}(x)+\sum_{k=0}^n\alpha_{n,k}^\ell P_k^{(\ell)}(x)
\end{align*}
and take the inner product with $x^kd\mu^{(\ell)}(x)$ on both sides, then one finds that
\begin{align*}
\langle xP_n^{(\ell+1)}(x),x^k\rangle_\ell=\langle P_n^{(\ell+1)}(x),x^k\rangle_{\ell+1}=0
\end{align*}
if $k<n$ according to \eqref{orl}. When $k=n$, it gives rise to
\begin{align*}
\alpha_{n,n}^\ell=\langle  P_n^{(\ell)}(x), xP_n^{(\ell+1)}(x)\rangle_\ell\langle P_n^{(\ell)}(x),P_n^{(\ell)}(x)\rangle_\ell^{-1}
=
H_{n}^{(\ell+1)}(H_{n}^{(\ell)})^{-1}.
\end{align*}
\end{proof}
\begin{remark}
This Geronimus transformation is a special case considered in \cite{garcia18} since we set the spectral parameter as zero here. If one considers a nonzero spectral parameter, a slightly different discrete non-abelian integrable system could be derived whose $\ell$ index acts as a discrete parameter and obey certain discrete dispersion relation. Since our start point in this part is to connect the non-abelian Toda lattice given by Gelfand et al with matrix orthogonal polynomials, we only need to consider the special case.
\end{remark}
The compatibility condition of Christoffel transformation \eqref{ct} and Geronimus transformation \eqref{gt} is
\begin{align*}
xP_n^{(\ell)}(x)=P_{n+1}^{(\ell)}(x)+(\mathcal{B}_n^\ell+\mathcal{B}_{n}^{\ell-1}-\mathcal{A}_{n+1}^{\ell-1})P_n^{(\ell)}(x)+\mathcal{B}^{\ell-1}_n(\mathcal{B}_{n-1}^{\ell}-\mathcal{A}_n^{\ell-1})P_{n-1}^{(\ell)}(x).
\end{align*}
Comparing it with the three term recurrence relation \eqref{ttrr2}, 
one finally gets
\begin{align}
H_{n+1}^{(\ell)}=H_n^{(\ell+2)}-H_n^{(\ell+1)}\left(
(H_n^{(\ell)})^{-1}-(H_{n-1}^{(\ell+2)})^{-1}
\right)H_n^{(\ell+1)}\label{ncdt}
\end{align}
which is a fully discrete non-commutative Toda lattice given by \cite[Thm 9.7.1]{gelfand05}.

The Lax pair of those discrete non-commutative Toda lattice equations can be formulated as follows. 
As discussed in last section, let's define $\Phi^\ell$ as a column vector whose elements are composed by $\{P_0^{(\ell)}(x),P_1^{(\ell)}(x),\cdots\}$, 
then three term recurrence relations \eqref{ttrr2} could be written as
\begin{align*}
x{\Phi}^\ell={\mathcal{L}}^\ell {\Phi}^\ell ,\quad {\mathcal{L}}^\ell=\left(\begin{array}{ccccc}
{a}_0^{(\ell)}&\mathbb{I}_p&&&\\
b_1^{(\ell)}&{a}_1^{(\ell)}&\mathbb{I}_p&&\\
&b_2^{(\ell)}&{a}_2^{(\ell)}&\mathbb{I}_p&\\
&&\ddots&\ddots&\ddots
\end{array}
\right).
\end{align*}
Moreover,
from the Christoffel and Geronimus transformations, one finds
\begin{align}\label{n1}
P_n^{(\ell)}(x)=P_n^{(\ell+1)}(x)+(\mathcal{B}_{n-1}^{\ell+1}-\mathcal{A}_n^\ell)P_{n-1}^{(\ell+1)}(x).
\end{align}
and thus
\begin{align*}
{\Phi}^\ell={\mathcal{M}}^{\ell+1}{\Phi}^{\ell+1},\quad {\mathcal{M}}^{\ell+1}=\left(
\begin{array}{cccc}
\mathbb{I}_p&&&\\
\mathcal{B}_0^{\ell+1}-\mathcal{A}_1^\ell&\mathbb{I}_p&&\\
&\mathcal{B}_1^{\ell+1}-\mathcal{A}_2^\ell&\mathbb{I}_p&\\
&&\ddots&\ddots\end{array}
\right).
\end{align*}
Therefore, the compatibility conditions could be given as
\begin{align*}
{\mathcal{M}}^{\ell+1}{\mathcal{L}}^{\ell+1}={\mathcal{L}}^{\ell}{\mathcal{M}}^{\ell+1}.
\end{align*}
Moreover, if one wants to find the nonlinear form satisfied by $\mathcal{A}_n^\ell$ and $\mathcal{B}_n^\ell$, then one could take the Geronimus transformation and equation \eqref{n1} as the discrete Lax pair. By rewriting Geronimus transformation as matrix form
\begin{align*}
x\Phi^{\ell+1}=\mathcal{J}^\ell \Phi^\ell,\quad \mathcal{J}^\ell=\left(\begin{array}{cccc}
\mathcal{B}_0^\ell&\mathbb{I}_p&&\\
&\mathcal{B}_1^\ell&\mathbb{I}_p&\\
&&\ddots&\ddots\end{array}
\right),
\end{align*}
and the compatibility condition is 
\begin{align*}
\mathcal{M}^{\ell+1}\mathcal{J}^\ell=\mathcal{J}^{\ell-1}M^\ell,
\end{align*}
from which one could get the nonlinear equation
\begin{align*}
\mathcal{B}_{n-1}^{\ell+1}-\mathcal{A}_n^\ell=\mathcal{B}_n^{\ell-1}-\mathcal{A}_{n+1}^{\ell-1},\quad
(\mathcal{B}_n^{\ell+1}-\mathcal{A}_{n+1}^\ell)\mathcal{B}_n^\ell=\mathcal{B}_{n+1}^{\ell-1}(\mathcal{B}_n^\ell-\mathcal{A}_{n+1}^{\ell-1}).
\end{align*}

At the end of this section, we remark that although each of $\{H_n^{(\ell)}\}_{\ell=0,1,\cdots}$ is a solution of non-commutative Toda equation \eqref{nctoda}, there are three different families of solutions connected by a single equation in equations \eqref{ncdt}, which is not a B\"acklund transformation. 
In following, we demonstrate a B\"acklund transformation of non-commutative Toda equation by using a reduction to the matrix orthogonal polynomials.

\section{A reductional case and B\"acklund transformation of non-abelian Toda lattice}\label{sec4}

In this part, we consider a B\"acklund transformation of non-abelian Toda lattice.  It was motivated by  a modification of moments in the scalar case, which was applied to connect Toda lattice with Volterra lattice \cite{peherstorfer07,chu08} and later to connect Camassa-Holm equation with a 2-component modified Camassa-Holm equation \cite{chang18}. We summarize the idea in the following diagram.\\

\begin{center}
\begin{tikzpicture}[font={ \small}]
\node (tran1) at (0,0) [draw, rectangle,minimum width=3cm, minimum height=1cm] {\begin{tabular}{c}standard orthogonal polynomials\end{tabular}};      
\node (gene1) at (0,-2) [draw, rectangle, align=left,minimum width=3cm,minimum height=1cm] {\begin{tabular}{c}symmetric orthogonal polynomials\end{tabular}};       
 \node(tran2) at (5,0) [draw, rectangle,minimum width=3cm, minimum height=1cm] {\begin{tabular}{c}Toda lattice\end{tabular}};      
 \node(gene2)  at (5,-2) [draw, rectangle, align=left,minimum width=3cm,minimum height=1cm] {\begin{tabular}{c}Volterra lattice\end{tabular}};
\draw[->] (tran1) -- (gene1);
\draw[->] (tran1) -- (tran2);
\draw[->] (gene1)--(gene2);
\draw[->](tran2)--(gene2);
\end{tikzpicture}
\end{center}

Such moment modification could be realized by setting odd moments to be zero and even moments to be non-zero elements, which is also referred to as unwrapping of measure. In the scalar case, 
one could consider a Cauchy transform of the measure, i.e.
\begin{align*}
F(z)=\int_{\mathbb{R}}\frac{d\mu(x)}{z-x}=\sum_{i=0}^\infty\frac{s_i}{z^{i+1}},\quad s_i=\int_{\mathbb{R}}x^i d\mu(x).
\end{align*}
Then a symmetric constraint on the measure leads to 
\begin{align*}
{s_{2i+1}}\ne 0,\quad s_{2i}=0
\end{align*}
which satisfy above-mentioned requirements and Hermite polynomial is an example. In the matrix orthogonal polynomial framework, one could take the matrix-valued Hermite polynomial for instance. In \cite{cafasso13}, Hermite-type matrix orthogonal polynomials were expressed by contour integrals and their connections with non-commutative Painlev\'e equation were studied. It was found that the Hermite-type matrix orthogonal polynomials satisfy a second-order differential equation, which is to admit a double integral representation of the matrix-valued Christoffel-Darboux kernels. Some other examples could be found at \cite{bertola21,ismail19}.

In our case, we consider a ``formal'' matrix-valued weight \footnote{The formal matrix-valued weight means that the weight function are not necessarily written in an explicit form since we don't need explicit expressions of those moments/polynomials but evolutions of the moments. For details, one could refer to \cite{brezinski80}.}
\begin{align}\label{dmu}
d\mu(x)=M(x)\exp(-x^2)dx,
\end{align}
where $M(x)\in\mathbb{R}^{p\times p}[x]$ and $M(x)=M^\top(x)$. We further require that $M(x)$ is an even function such that $M(x)=M(-x)$. This measure is obviously symmetric so that moments
\begin{align*}
\int_{\mathbb{R}}x^n d\mu(x)=0
\end{align*}
when ($n$ mod 2)=1 and 
\begin{align*}
\int_{\mathbb{R}}x^n d\mu(x):=d_{n/2}
\end{align*}
when ($n$ mod 2)=0. 
Moreover, those symmetric matrix orthogonal polynomials $\{Q_n(x)\}_{n\in\mathbb{N}}$, which are orthogonal with $d\mu(x)$ in \eqref{dmu}, admit the following quasi-determinant forms.
\begin{proposition}
The even and odd order symmetric matrix orthogonal polynomials are expressed by
\begin{align*}
Q_{2n}(x)=\left|
\begin{array}{cccc}
d_0&\cdots&d_{n-1}&d_n\\
\vdots&&\vdots&\vdots\\
d_{n-1}&\cdots&d_{2n-2}&d_{2n-1}\\
\mathbb{I}_p&\cdots&x^{2n-2}\mathbb{I}_p&\boxed{x^{2n}\mathbb{I}_p}\end{array}
\right|,\quad Q_{2n+1}(x)=x\left|\begin{array}{cccc}
d_1&\cdots&d_n&d_{n+1}\\
\vdots&&\vdots&\vdots\\
d_n&\cdots&d_{2n}&d_{2n+1}\\
\mathbb{I}_p&\cdots&x^{2n-2}\mathbb{I}_p&\boxed{x^{2n}\mathbb{I}_p}\end{array}
\right|
\end{align*}
with orthogonality
\begin{align*}
\langle Q_{2n}(x),Q_{2m+1}(x)\rangle=0,\quad \langle Q_{2n}(x),Q_{2m}(x)\rangle=H_n^{(0)}\delta_{n,m},\quad  \langle Q_{2n+1}(x),Q_{2m+1}(x)\rangle=H_n^{(1)}\delta_{n,m},
\end{align*}
where 
\begin{align*}
H_n^{(k)}=\left|\begin{array}{cccc}
d_k&\cdots&d_{k+n}\\
\vdots&&\vdots\\
d_{k+n}&\cdots&\boxed{d_{k+2n}}\end{array}
\right|.
\end{align*}
\end{proposition}
\begin{proof}
Here we give expressions for even order polynomials,  which are similar to the odd ones. Thus we omit odd ones here. By rearranging rows and columns of even and odd orders, we notice that 
\begin{align*}
Q_{2n}(x)&=\left|\begin{array}{cccccc}
d_0&0&\cdots&d_{n-1}&0&d_n\\
0&d_1&\cdots&0&d_n&0\\
\vdots&\vdots&&\vdots&\vdots&\vdots\\
d_{n-1}&0&\cdots&d_{2n-2}&0&d_{2n-1}\\
0&d_n&\cdots&0&d_{2n-1}&0\\
\mathbb{I}_p&x\mathbb{I}_p&\cdots&x^{2n-2}\mathbb{I}_p&x^{2n-1}\mathbb{I}_p&\boxed{x^{2n}\mathbb{I}_p}\end{array}
\right|\\
&=x^{2n}\mathbb{I}_p-\left(
\mathbb{I}_p,\,\cdots,\,x^{2n-2}\mathbb{I}_p,\,x\mathbb{I}_p,\cdots,\,x^{2n-1}\mathbb{I}_p
\right)\left(\begin{array}{cc}
A_n^{(0)}&0\\
0&A_n^{(1)}\end{array}
\right)^{-1}\left(\begin{array}{c}
v\\0\end{array}
\right),
\end{align*}
where $A_n^{(k)}=\left(
d_{i+j+k}
\right)_{i,j=0}^{n-1}$ and $v=(d_n,\cdots,d_{2n-1})$. Moreover, since the last term in the above equation is actually $(\mathbb{I}_p,\cdots,x^{2n-2}\mathbb{I}_p)(A_n^{(0)})^{-1}v$, one is led to the expression of $Q_{2n}(x)$.

On the other hand, the orthogonality is easily verified by taking the advantage of the expansion of quasi-determinant, and we omit it here.
\end{proof}

Compared with expressions of matrix orthogonal polynomials $\{P^{(\ell)}_n(x)\}_{n\in\mathbb{N}}$, we know that 
\begin{align*}
\text{$Q_{2n}(x)=P^{(0)}_n(x^2)$, \quad $Q_{2n+1}(x)=xP_n^{(1)}(x^2)$ }
\end{align*}
if moments $\{d_i\}_{i\in\mathbb{N}}$ coincide with $\{m_i\}_{i\in\mathbb{N}}$.
However, symmetric matrix orthogonal polynomials admit special properties. For example,
the modification of measure leads to reduced three term recurrence relations.

\begin{proposition}\label{prop4}
Regarding with matrix orthogonal polynomials $\{Q_n(x)\}_{n\in\mathbb{N}}$ with symmetric measure, there are reduced three term recurrences
\begin{align}\label{lplv1}
xQ_{2n}(x)=Q_{2n+1}(x)+\xi_nQ_{2n-1}(x),\quad xQ_{2n+1}(x)=Q_{2n+2}(x)+\zeta_{n+1}Q_{2n}(x),
\end{align}
where $\xi_{n}=H_{n}^{(0)}\left(H_{n-1}^{(1)}\right)^{-1}$ and $\zeta_{n+1}=H_{n}^{(1)}\left(H_{n}^{(0)}
\right)^{-1}$. 
\end{proposition}
\begin{proof}
We prove the even case here. By noting that $xQ_{2n}(x)$ is monic and it is composed by odd order polynomials, we know that
\begin{align*}
xQ_{2n}(x)=Q_{2n+1}(x)+\sum_{k=0}^{n-1}\xi_{2n,k}Q_{2k+1}(x),\quad \xi_{2n,k}=\langle xQ_{2n}(x),Q_{2k+1}(x)\rangle (H_{k}^{(1)})^{-1}.
\end{align*}
Moreover, by making use of the orthogonality, we know that only $k=n-1$ contributes and the result is easily obtained.
\end{proof}

The three term recurrence relation could be written as the following matrix form
\begin{align*}
x\left(\begin{array}{c}
Q_0(x)\\
Q_1(x)\\
Q_2(x)\\
\vdots\end{array}
\right)=\left(\begin{array}{ccccc}
0&1&&&\\
\zeta_1&0&1&&\\
&\xi_1&0&1&\\
&&\ddots&\ddots&\ddots\end{array}
\right)\left(\begin{array}{c}
Q_0(x)\\
Q_1(x)\\
Q_2(x)\\
\vdots\end{array}
\right)
\end{align*}
If we denote $Q=(Q_0(x),Q_1(x),\cdots)^\top$, then the above matrix form is equivalent to
\begin{align*}
xQ=(\Lambda+\Lambda^{-1}\gamma)Q:=\mathcal{L}Q,
\end{align*}
where $\gamma=\text{diag}(\zeta_1,\xi_1,\zeta_2,\xi_2,\cdots)$.

Let's impose a time flow on the measure such that $\partial_td\mu(x;t)=x^2d\mu(x;t)$, from which we could get a general $d\mu(x;t)$ and $d\mu(x;t)=\exp(x^2 t)d\mu(x)$ (c.f. Sec. \ref{ho} but note that odd flows $t_{2j+1}$ are trivial flows in this case). Such an assumption is called a generalized Freud weight and examples in scalar case were considered in \cite{aptekarev97,clarkson16}. Furthermore, with this assumption, moments $\{d_i\}_{i\in\mathbb{N}}$ satisfy time evolutions $\partial_t d_i=d_{i+1}$ and $\{H_n^{(\ell)},\ell=0,1\}_{n\in\mathbb{N}}$ in this case satisfy the non-abelian Toda equation separately. Moreover, when it comes to those reduced matrix orthogonal polynomials with symmetric measure $\{Q_n(x;t)\}_{n\in\mathbb{N}}$, one can state the following proposition.
\begin{proposition}
The time flow of $\{Q_n(x;t)\}_{n\in\mathbb{N}}$ satisfies
\begin{align}\label{lplv2}
\partial_t Q_{2n}(x;t)=\alpha_nQ_{2n-2}(x;t),\quad \partial_t Q_{2n+1}(x;t)=\beta_nQ_{2n-1}(x;t),
\end{align}
where $\alpha_n=-H_{n}^{(0)}\left(
H_{n-1}^{(0)}
\right)^{-1}$ and $\beta_n=-H_{n}^{(1)}\left(
H_{n-1}^{(1)}
\right)^{-1}$.
\end{proposition}
\begin{proof}
By introducing the time flow, one can assume
\begin{align*}
\partial_t Q_{2n}(x;t)=\sum_{j=0}^{2n-1} \xi_{n,j}Q_j(x;t).
\end{align*}
Taking $t$-derivative to the orthogonal relation $\langle Q_{2n},Q_j\rangle=0$ for $j<2n$, one gets 
\begin{align*}
\langle\partial_t Q_{2n},Q_j\rangle+\langle x^2Q_{2n},Q_j\rangle=0.
\end{align*}
From the orthogonality of $\{Q_n(x)\}_{n\in\mathbb{N}}$, one knows that only $j=2n-2$ contributes, and 
\begin{align*}
\xi_{n,2n-2}=-\langle Q_{2n},x^2Q_{2n-2}\rangle\langle Q_{2n-2},Q_{2n-2}\rangle^{-1}=-H_{n}^{(0)}\left(
H_{n-1}^{(0)}
\right)^{-1}.
\end{align*}
The odd case could be similarly proved.
\end{proof}
The coefficients $\{\alpha_n,\beta_n\}_{n\in\mathbb{N}}$ are connected with $\{\xi_n,\zeta_n\}_{n\in\mathbb{N}}$ by
\begin{align*}
\alpha_n=-\xi_n\zeta_n,\quad \beta_n=-\zeta_{n+1}\xi_n
\end{align*}
and the Lax pair of non-commutative Volterra lattice is then given by
\begin{align*}
\partial_t Q=-(\mathcal{L}^2)_-Q.
\end{align*}
The compatibility condition gives rise to the Lax representation
\begin{align*}
\partial_t \mathcal{L}=[\mathcal{L},(\mathcal{L}^2)_-],\quad \mathcal{L}=\Lambda+\Lambda^{-1}\gamma,
\end{align*}
and the nonlinear form of non-commutative Volterra lattice could be written as
\begin{align*}
\partial_t \xi_n=\zeta_{n+1}\xi_n-\xi_n\zeta_n,\quad \partial_t \zeta_{n+1}=\xi_{n+1}\zeta_{n+1}-\zeta_{n+1}\xi_n.
\end{align*}
By substituting $\xi_n=H_n^{(0)}\left(
H_{n-1}^{(1)}
\right)^{-1}$ and $\zeta_{n+1}=H_n^{(1)}\left(
H_n^{(0)}
\right)^{-1}$, one obtains a bilinear form
\begin{align*}
&(H_n^{(0)})^{-1}\partial_t H_n^{(0)}-(H_{n-1}^{(1)})^{-1}\partial_t H_{n-1}^{(1)}=(H_n^{(0)})^{-1}H_n^{(1)}-(H_{n-1}^{(0)})^{-1}H_{n-1}^{(1)},\\
&(H_n^{(1)})^{-1}\partial_t H_n^{(1)}-(H_n^{(0)})^{-1}\partial_t H_n^{(0)}=(H_n^{(1)})^{-1}H_{n+1}^{(0)}-(H_{n-1}^{(1)})^{-1}H_n^{(0)}.
\end{align*}
We call it the bilinear form since these two equations could be verified by non-commutative Jacobi identity directly.
Obviously, those two equations connecting two families of solutions of non-commutative Toda lattice, and thus an auto-B\"acklund transformation of the non-commutative Toda lattice.
\begin{remark}
 In fact, in \cite[eq. (5)-(6)]{popowicz83}, the author proposed the same B\"acklund transformation for non-commutative Toda lattice in $2+1$ dimension and in \cite{o'connell20}, it was interpreted as a diffusion with generator and a stochastic version of a series of B\"acklund transformations between non-commutative Toda systems with different numbers of particles. Moreover, if we denote $\gamma_{2n-1}=\zeta_n$ and $\gamma_{2n}=\xi_n$, the nonlinear form could be written in a unified way $\partial_t \gamma_n=\gamma_{n+1}\gamma_n-\gamma_n\gamma_{n-1}$, which could be obtained from equation \eqref{t2flow} by taking $a_n=0$.  In an earlier work of Moser \cite{moser75}, it was pointed out that if one considers an invariant manifold on Toda lattice, then Volterra lattice appears. This idea holds true for the non-commutative case.
 \end{remark}
We conclude our results as follows.
\begin{proposition}\label{prop5}
The non-commutative Volterra lattice
\begin{align}\label{nclv}
\partial_t \gamma_{n}=\gamma_{n+1}\gamma_{n}-\gamma_{n}\gamma_{n-1}
\end{align}
has a Lax pair $\partial_t\mathcal{L}=[\mathcal{L},(\mathcal{L}^2)_-]$ with Lax matrix
\begin{align*}
\mathcal{L}=\left(\begin{array}{ccccc}
0&1&&&\\
\gamma_1&0&1&&\\
&\gamma_2&0&1&\\
&&\ddots&\ddots&\ddots
\end{array}
\right).
\end{align*}
Moreover, solutions of the non-abelian Volterra lattice are given by 
\begin{align*}
\gamma_{2n}=H_{n}^{(0)}\left(H_{n-1}^{(1)}\right)^{-1},\quad \gamma_{2n+1}=H_{n}^{(1)}\left(H_{n}^{(0)}
\right)^{-1},
\end{align*}
where $\{H_n^{(\ell)},\ell=0,1\}_{n\in\mathbb{N}}$ are Hankel quasi-determinants expressed by
\begin{align*}
H_n^{(\ell)}=\left|\begin{array}{ccc}
d_\ell&\cdots&d_{\ell+n}\\
\vdots&&\vdots\\
d_{\ell+n}&\cdots&\boxed{d_{\ell+2n}}\end{array}
\right|
\end{align*}
with time evolution $\partial_t d_i=d_{i+1}$.
\end{proposition}
\begin{remark}
The non-commutative Volterra lattice could be regarded as a discrete version of non-commutative KdV equation. Taking 
\begin{align*}
\gamma_n=1+\epsilon^2 r, \quad \gamma_{n\pm 1}=1+\epsilon^2 r(x\pm\epsilon),
\end{align*}
then equation \eqref{nclv} could be written as
\begin{align*}
r_t=2\epsilon r_x+\epsilon^3 (r_xr+rr_x)+\frac{\epsilon^3}{3} r_{xxx}+o(\epsilon^4),
\end{align*}
which is the non-commutative KdV equation \cite{dimakis00,hamanaka07} after the scaling $T=\epsilon^3 t$ and $X=x-2\epsilon t$ and taking the limit $\epsilon\to0$.
\end{remark}
One could consider a non-abelian Volterra hierarchy in analog of Toda hierarchy. Since odd order flows are trivial in this case, i.e. $\partial_{t_{2n+1}}Q_j(x;t)=0$ for arbitrary $j$, we only consider even order flows. Similar with the proof in Proposition \ref{hof}, one could compute that
\begin{align*}
\partial_{t_{2n}}Q=-(\mathcal{L}^{2n})_-Q,
\end{align*}
and the compatibility condition is 
\begin{align*}
\partial_{t_{2n}}\mathcal{L}=[\mathcal{L},(\mathcal{L}^n)_-],
\end{align*}
which is a discretization of the non-commutative KdV hierarchy.

\section{Concluding remarks}
Tau function is a fundamental viewpoint in integrable theory and it links different subjects in modern mathematical physics. In this paper, we focus on the ``non-commutative tau functions'' expressed by quasi-determinants, which in fact are ratios of tau functions in scalar case. Normalization factors of matrix orthogonal polynomials acting as tau functions are used to formulate several non-abelian integrable systems.
Although we find that those normalization factors satisfy certain non-abelian integrable lattices, it is still unknown to us whether there are linear differential systems or Virasoro constraints to characterize those tau functions.

Besides, it is also mysterious to know whether there are any generalizations of $B_\infty$ or $C_\infty$-types of Toda lattices into the non-commutative case. We hope to investigate those non-commutative integrable lattices by making use of matrix Cauchy bi-orthogonal polynomials and matrix skew-orthogonal polynomials. Along this way, technique of quasi-Pfaffian needs to be largely explored. Those problems will be discussed in some continued works.

\section*{Acknowledgement}
The author would like to thank Prof. Chunxia Li, Prof Jacek Szmigielski and Prof. Guo-Fu Yu for helpful discussions and suggestions, and thanks Prof. Vladimir Retakh for bringing me some useful references. This work is partially funded by grants (NSFC12101432, NSFC11971322).


\begin{thebibliography}{1}

\bibitem{adler99}
M. Adler, E. Horozov and P. van Moerbeke.
\newblock The Pfaff lattice and skew-orthogonal polynomials.
\newblock \emph{Int. Math Res. Not.}, 11 (1999), 569-588.

\bibitem{adler95}
M. Adler and P. van Moerbeke.
\newblock{Matrix integrals, Toda symmetries, Virasoro constraints, and orthogonal polynomials}.
\newblock \emph{Duke Math J.}, 80 (1995), 863-911.

\bibitem{adler03}
M. Adler and P. van Moerbeke.
\newblock{Recursion relations for unitary integrals, combinatorics and the Toeplitz lattice}.
\newblock {\em Comm. Math. Phys.}, 237 (2003), 397-440.

\bibitem{adler09}
M. Adler, P. van Moerbeke and P. Vanhaecke.
\newblock {Moment matrices and multi-component KP, with applications to random matrix theory}.
\newblock {\em Commun. Math. Phys.,} 286 (2009), 1-38.

\bibitem{alvarez17}
C. \'Alvarez-Fern\'andez, G. Ariznabarreta, J. Gar\'cia-Ardila, M. Ma\~nas and F. Marcell\'an.
\newblock{Christoffel transformations for matrix orthogonal polynomials in the real line and the non-Abelian 2D Toda lattice hierarchy}.
\newblock \emph{Int. Math. Res. Not.}, (2017), 1285-1341.


\bibitem{aptekarev97}
A. Aptekarev, A. Branquinho and F. Marcell\'an.
\newblock{Toda-type differential equations for the recurrence coefficients of orthogonal polynomials in Freud transformation}.
\newblock \emph{J. Comp. Appl. Math.}, 78 (1997), 139-160.


\bibitem{ari14}
G. Ariznabarreta and M. Ma\~nas.
\newblock  Multivariate orthogonal polynomials and integrable systems.
\newblock \emph{Adv. Math.}, 302 (2014), 628-739.

\bibitem{bertola21}
M. Bertola.
\newblock Abelianization of matrix orthogonal polynomials.
\newblock arXiv: 2107.12998.

\bibitem{brezinski80}
C. Brezinski.
\newblock Pad\'e-type approximation and general orthogonal polynomials.
\newblock International Series of Numerical Mathematics, 50. \emph{Birkkh\"auser Verlag, Basel-Boston, Mass.,} 1980.

\bibitem{brezinski90}
C. Brezinski.
\newblock A direct proof of the Christoffel-Darboux identity and its equivalence to the recurrence relationship.
\newblock \emph{J. Comp. Appl. Math.}, 32 (1990), 17-25.

\bibitem{branquinho20}
A. Branquinho, A. Foulqui\'e-Moreno and J. Garci\'a-Ardila.
\newblock{Matrix Toda and Volterra lattices}.
\newblock \emph{Appl. Math. Comp.}, 365 (2020), 124722.

\bibitem{cafasso13}
M. Cafasso and D. Manuel.
\newblock Non-commutative Painlev\'e equations and Hermite-type matrix orthogonal polynomials.
\newblock \emph{Commun. Math. Phys.}, 326 (2014), 559-583.

\bibitem{chang182}
X. Chang, Y. He, X. Hu and S. Li.
\newblock  Partial-skew-orthogonal polynomials and related integrable lattice with Pfaffian tau-functions.
\newblock \emph{Commun. Math. Phys.}, 364 (2018), 1069-1119.

\bibitem{chang18}
X. Chang, X. Hu and S. Li.
\newblock Moment modification, multipeakons, and nonisospectral generalizations.
\newblock \emph{J. Diff. Equa.}, 265 (2018), 3858-3887.



\bibitem{chu08}
M. Chu.
\newblock Linear algebra algorithms as dynamical systems.
\newblock \emph{Acta Numerica}, 17 (2008), 1-86.

\bibitem{clarkson16}
P. Clarkson, K. Jordaan and A. Kelil.
\newblock A generalized Freud weight.
\newblock \emph{Stud. Appl. Math.}, 136 (2016), 288-320.


\bibitem{o'connell13}
N. O'Connell.
\newblock{Geometric RSK and the Toda lattice}.
\newblock \emph{Illinious J. Math.}, 57 (2013), 883-918.

\bibitem{o'connell20}
N. O'Connell.
\newblock{Interacting diffusions on positive definite matrices}.
\newblock \emph{Prob. Theo. Rel. Fiel.}, 180 (2021), 679-726.

\bibitem{damanik08}
D. Damanik, A. Pushnitski and B. Simon.
\newblock{The analytic theory of matrix orthogonal polynomials.}
\newblock arXiv: 0711.2703.

\bibitem{deift03}
P. Deift.
\newblock{Orthogonal Polynomials and Random Matrices: A Riemann-Hilbert Approach}.
\newblock Courant Lecture Notes 3, American Mathematical Society, 2000.

\bibitem{di11}
P. Di Francesco and R. Kedem.
\newblock{Non-commutative integrability, paths and quasi-determinants}.
\newblock \emph{Adv. Math.}, 10 (2011), 97-152.


\bibitem{dimakis00}
A. Dimakis and F. Muller-Hoissen.
\newblock{non-commutative Korteweg-de-Vries equation}.
\newblock Phys. Lett. A, 278 (2000), 139.

\bibitem{duits20}
M. Duits and A. Kuijlaars.
\newblock{The two periodic Aztec diamond and matrix valued orthogonal polynomials}.
\newblock to appear in \emph{JEMS}, arXiv: 1712.05636.

\bibitem{etingof98}
P. Etingof, I. Gelfand and V. Retakh.
\newblock{
Nonabelian Integrable Systems, Quasideterminants, and Marchenko Lemma, 
}
\newblock \emph{Math. Research Lett.}, 5 (1998), 1-12.


\bibitem{forrester10}
P. Forrester.
\newblock {Log-gases and random matrices}.
\newblock  {London Mathematical Society Monographs Series 34.}
\newblock \emph{ Princeton University Press, Princeton, NJ, } 2010.

\bibitem{garcia18}
J. Garc\'ia-Ardila, L. Garza and F. Marcell\'an.
\newblock {A canonical Geronimus transformation for matrix orthogonal polynomials}.
\newblock \emph{Lin. and Mult. Alg.}, 66 (2018), 357-381.

\bibitem{gelfand05}
I. Gelfand, S. Gelfand, V. Retakh and R. Wilson.
\newblock{quasi-determinants}.
\newblock \emph{Adv. Math.}, 193 (2005), 56-141.

\bibitem{gelfand95}
I. Gelfand, D. Krob, A. Lascoux, B. Leclerc, V. Retakh and J. Thibon.
\newblock Noncommutative symmetric functions.
\newblock \emph{Adv. Math.}, 112 (1995), 218-348.


\bibitem{gerasimov91}
A. Gerasimov, A. Marshakov, A. Mironov, A. Morozov and A. Orlov.
\newblock{Matrix models of two-dimensional gravity and Toda theory}.
\newblock \emph{Nucl. Phys. B}, 357 (1991), 565-618.

\bibitem{gilson07}
C. Gilson and J. Nimmo.
\newblock {On a direct approach to quasi-determinant solutions of a non-commutative KP equation}.
\newblock \emph{J. Phys. A: Math. Theor.}, 40 (2007), 3839.

\bibitem{gilson072}
C. Gilson, J. Nimmon and Y. Ohta.
\newblock{quasi-determinant solutions of a non-Abelian Hirota-Miwa equation.}
\newblock \emph{J. Phys. A: Math. Theor.}, 40 (2007), 12607.

\bibitem{grunbaum08}
F. Gr\"unbaum and M. de la Iglesia.
\newblock{Matrix valued orthogonal polynomials arising from group representation theory and a family of quasi-birth-and-death processes.}
\newblock \emph{SIAM J. Mat. Anal. Appl.}, 30 (2008), 741-761.

\bibitem{hamanaka07}
M. Hamanaka.
\newblock {Notes on exact multi-soliton solutions of non-commutative integrable hierarchies.}
\newblock JHEP, 02 (2007), 094.

\bibitem{ismail19}
M. Ismail, E. Keolimk and P. Rom\'an.
\newblock {Matrix valued Hermite polynomial, Burchnall formulas and non-abelian Toda lattice}.
\newblock \emph{Adv. Appl. Math.,} 110 (2019), 235-269.

\bibitem{krein49}
M. Krein.
\newblock {Infinite J-matrices and a matrix-moment problem}.
\newblock \emph{Dokl. Akad. Nauk. SSSR}, 69 (1949), 125-128.

\bibitem{krob95}
D. Krob and B. Leclerc.
\newblock Minor identities for quasi-determinants and quantum determinants.
\newblock \emph{Commun. Math. Phys.}, 169 (1995), 1-23.


\bibitem{li08}
C. Li and J. Nimmo.
\newblock{quasi-determinant solutions of a non-Abelian Toda lattice and kink solutions of a matrix sine-Gordon equation.}
\newblock \emph{Proc. R. Soc. A.},  464 (2008), 951-966.

\bibitem{li20}
S. Li and G. Yu.
\newblock{Christoffel transformations for (partial-)skew-orthogonal polynomials and applications}.
\newblock arXiv: 2008.00273.

\bibitem{lundmark05}
H. Lundmark and J. Szmigieski.
\newblock{Degasperis-Procesi peakons and the discrete cubic string.}
\newblock \emph{Int. Math. Res. Pap.}, 2 (2005), 53-116.

\bibitem{miranian05}
L. Miranian.
\newblock{Matrix valued orthogonal polynomials}.
\newblock PhD thesis, University of California, Berkeley. ProQuest Dissertations Publishing, 2005. 3196602.

\bibitem{moser75}
J. Moser.
\newblock {Finitely many mass points on the line under the influence of an exponential potential---An integrable system.}
\newblock \emph{Dynamical systems, theory and applications}, pp. 467-497, Lecture Notes in Phys., Vol. 38, Springer, Berlin, 1975.


\bibitem{peherstorfer07}
F. Peherstorfer, V. Spiridonov and A. Zhedanov.
\newblock{Toda chain, Stieltjes function, and
orthogonal polynomials}.
\newblock \emph{Theor. Math. Phys.}, 151 (2007), 505-528.

\bibitem{popowicz83}
Z. Popowicz.
\newblock {The generalized non-abelian Toda lattice}.
\newblock \emph{Z. Phys. C}, 19 (1983), 79-81.

\bibitem{retakh10}
V. Retakh and V. Rubtsov.
\newblock non-commutative Toda chains, Hankel quasi-determinants and Painlev\'e equation.
\newblock \emph{J. Phys. A: Math. Theor.}, 43 (2010), 505204.

\bibitem{sinap96}
A. Sinap and W. van Assche. 
\newblock{Orthogonal matrix polynomials and applications.}
\newblock \emph{J. Comp. Appl. Math.}, 66 (1996), 27-52.

\bibitem{spicer11}
P. Spicer, F. Nijhoff and P. van der Kamp.
\newblock {Higher analogues of the discrete-time Toda equation and the quotient-difference algorithm}.
\newblock \emph{Nonlinearity}, 24 (2011), 2229.


\end{thebibliography}
\end{document}